\documentclass[conference]{IEEEtran}
\IEEEoverridecommandlockouts
\usepackage{cite}
\usepackage{amsmath,amssymb,amsfonts}
\usepackage{graphicx}
\usepackage{textcomp}

\usepackage{multirow}
\usepackage{subcaption}
\usepackage{amsthm}
\usepackage{algorithm}
\usepackage[noend]{algorithmic}

\usepackage{cuted}

\theoremstyle{plain}
\newtheorem{thm}{Theorem} 
\theoremstyle{definition}
\newtheorem{defn}[thm]{Definition} 
\newtheorem{theorem}{Theorem}

\theoremstyle{definition}
\newtheorem{observation}[theorem]{Observation}

\usepackage{soul}
\usepackage{xcolor}

\usepackage{amsmath}
\usepackage{stfloats}
\usepackage{blindtext}
\usepackage{float}

\setstcolor{red}

\DeclareMathAlphabet\mathbfcal{OMS}{cmsy}{b}{n}

\renewcommand\thesubsubsection{\thesubsection.\arabic{subsubsection}}

\def\BibTeX{{\rm B\kern-.05em{\sc i\kern-.025em b}\kern-.08em
    T\kern-.1667em\lower.7ex\hbox{E}\kern-.125emX}}
    
\begin{document}

\title{Truthful Computation Offloading Mechanisms for Edge Computing}

\author{\IEEEauthorblockN{Weibin Ma}
\IEEEauthorblockA{Department of Computer and \\Information Sciences\\
University of Delaware\\
Newark, Delaware, USA\\
Email: weibinma@udel.edu}
\and
\IEEEauthorblockN{Lena Mashayekhy}
\IEEEauthorblockA{Department of Computer and \\Information Sciences\\
University of Delaware\\
Newark, Delaware, USA\\
Email: mlena@udel.edu}
}

\maketitle

\begin{abstract}
 Edge computing (EC) is a promising paradigm
providing a distributed computing solution for users at the
edge of the network. Preserving satisfactory quality of experience (QoE)
for users when offloading their computation to EC is a non-trivial  problem. 
Computation offloading in EC requires
jointly optimizing access points (APs)  allocation and edge service placement for  users, which  is 
computationally intractable due to its combinatorial nature. 
Moreover, users are self-interested, and they can misreport  their preferences 
leading to an inefficient resource allocation and network congestion. 
In this paper, we tackle this problem and design a novel mechanism based on algorithmic mechanism design to implement a system equilibrium.
Our mechanism assigns a proper pair of AP and edge server along with a service price for each new
joining user maximizing the instant social surplus while satisfying
all users' preferences in the EC system. 
Declaring true preferences is a weakly dominant strategy for the
users. 
The experimental results show that  our mechanism  outperforms user equilibrium and
random selection strategies in terms of the experienced end-to-end latency. 
\end{abstract}
	

\section{Introduction}\label{sec:intr}
With the explosive growth of smart devices, a bulk of computationally intensive applications, as exemplified by face recognition, online gaming, and video streaming, are becoming prevalent. 
However, the smart devices  possess limited resources (e.g., limited computation capabilities and battery lifetime), which may lead to unsatisfactory computation experience. The commonly used approach is to offload computational tasks to a powerful cloud platform~\cite{chen2014decentralized}. 
However, the long distance between  devices and the cloud will cause a significant increase in delay and network congestion.

To overcome this challenge, edge computing (EC) has recently been introduced as an emerging solution that enables offloading  computational tasks to the physically proximal EC mini-datacenters, called cloudlets~\cite{satyanarayanan2017emergence,Farhangi:ICFEC20}. 
While EC brings many opportunities to guarantee quality of experience (QoE) for users, 
new challenges arise due to the restricted coverage of cloudlets and their limited computational resources.  
To maintain the QoE of users, designing efficient realtime computation offloading  is hence becoming crucial in edge computing. 
The computation offloading problem consists of jointly optimizing access points (APs)  allocation and edge service placement for  EC users, which  is  computationally intractable due to its combinatorial nature.

In this paper, we design a novel mechanism called computation offloading and pricing mechanism (COPM) 
 to satisfy QoE of each joining user by meeting its application-specific end-to-end latency requirements. 
The goal of our proposed mechanism is to maximize the  instant social surplus, which is defined as the sum of the valuation of the new user and the system.
To tackle the complexity of COPM, we then propose an online offloading mechanism, called DAPA.
When a new user requests an edge service at any time, DAPA collects current information of the system and then 
 assigns an optimal decision pair (the best AP for connection and the best edge server for computation) to the user. 
 It also determines the user's   corresponding payment for the edge service. If  no feasible solution exists for this user, DAPA  suggests the new user offloading its task to the remote cloud. 

When a new user requests an edge service, it will report its 
maximum tolerable end-to-end latency in order to receive the best decision pair
to offload and complete its task. A user may misreport this value to increase its own utility.
Such an action  could inversely decrease the overall system efficiency. 
Therefore, designing an \textit{incentive-compatible} (or truthful) mechanism in which  users have no incentive to lie about their true preferences is extremely 
important for achieving system efficiency and implementing a system equilibrium. 
Our goal is to design an efficient incentive-compatible mechanism to determine an optimal decision pair with a corresponding payment for each user  satisfying their QoE requirements while maximizing the  instant social surplus. 
To the best of our knowledge, this is the first work that  simultaneously optimizes online AP allocation, service placement, and pricing of computation offloading  by utilizing algorithmic mechanism design.
Our proposed mechanism implements a weakly dominant strategy equilibrium for  users.

	The rest of the paper is organized as follows. Section~\ref{sec:related} reviews related work. The system model is described in Section~\ref{sec:model}. The problem formulation and COPM mechanism are presented in Section~\ref{sec:mecha_design}. In Section~\ref{sec:algorithm}, we describe  our efficient online algorithmic solution. Performance evaluation is carried out in Section~\ref{sec:results}. Section~\ref{sec:conclusions} concludes the paper.

\section{Related Work}\label{sec:related} 
In the presence of multiple  
 cloudlets, resource management becomes extremely important as it directly impacts edge service quality and system efficiency. 
Xu et al.~\cite{xu2015efficient} formulated a  capacitated cloudlet placement problem to minimize the average transmission delay between  users and cloudlets and proposed an  approximation algorithm to solve it. 
Jia et al.~\cite{jia2016cloudlet} studied the load balancing problem among multiple  cloudlets. Bhatta and Mashayekhy~\cite{bhatta2019cost} proposed a heuristic cost-aware cloudlet placement approach that 
guarantees minimum latency for edge services. 
Wang et al.~\cite{wang2017online} formulated the dynamic resource allocation problem in edge computing considering user mobility and proposed an online algorithm to solve it by decoupling the problem into a series of solvable sub-problems. 
However, none of these studies considers the selfish behavior of the users.

Game theory  has been widely used to model and analyze different allocation problems. 
	Algorithmic mechanism design  deals with efficiently-computable
algorithmic solutions in the presence of strategic players who  may misreport their input, and it has been used in distributed computing \cite{mashayekhy2015truthful,shi2014online,sharghivand2018qos}.
	Zavodovski et al.~\cite{zavodovski2019decloud}
proposed  an incentive compatible double auction mechanism, called DeCloud,  to  offer  pay-as-you-go
edge services, where ad hoc clouds can be spontaneously formed on the edge of the network. 	
Kiani and Ansari~\cite{kiani2017toward} proposed 
a revenue-maximizing auction-based mechanism for edge computing resources. However, the mechanism is not  incentive compatible. 
Ma~et al.~\cite{ma2019cyclic} modeled the resource allocation problem as a three-sided cyclic game (3CG), where  edge nodes and service providers  cooperate for completing user requests and compete for their
own interest. 3CG is proved to have pure-strategy Nash equilibria 
and an approximation ratio.

Nevertheless, none of the existing work jointly addresses the AP allocation and service placement problem along with determining service pricing  in the EC system. In this paper, we propose an  online incentive-compatible computation offloading mechanism  to address this problem. 

\section{Edge Computing System Model}\label{sec:model}
We consider an EC system with a set of cloudlets, each of which is equipped with an AP (e.g., base station or WiFi hotspots) and edge servers, to provide edge services for users (Fig.~\ref{fig:framework}). 
A regional cloudlet (or a group of cloudlets) can act as the EC coordinator 
with the responsibility of collecting system
information such as the user requests and system status.
We denote a set of cloudlets by $\mathcal{M}=\{1,2,\dots,M\}$ and  a set of users by $\mathcal{N}=\{1,2,\dots,N\}$. 
Users join and leave the system dynamically. 
Each cloudlet~$j \in \mathcal{M}$ has one or multiple edge servers with  
computation capability~$\mathbf{F}_j(\tau)$ (i.e., CPU cycles per second)
and memory capacity~$\mathbf{D}_j(\tau)$ at time~$\tau$. 
Each AP~$i \in \mathcal{M}$ can provide service to~$\mathbf{P}_i$ users simultaneously and 
has a bandwidth~$\mathbf{B}_i(\tau)$. 
The cloudlets are interconnected by a wired network (e.g., wide-area network (WAN) or local-area network (LAN)).

Each user has a computational task requiring remote execution (a user can have multiple  tasks, and each is treated independently in this system).  The task of user $k \in \mathcal{N}$ is defined by $(C_k, D_k, T_k)$, where $C_k$ represents the total amount of computational cycles required to obtain the outcome of the task, $D_k$ denotes the data size of the task, and $T_k$ is the maximum tolerable end-to-end latency, measured in time units,  for completing the  task. 
Each user can be  connected to a cloudlet via  an AP through a wireless communication (e.g., WiFi, 4G, or 5G) to offload a task.

\begin{figure}[t]
	\centering
	\includegraphics[width=0.43\textwidth]{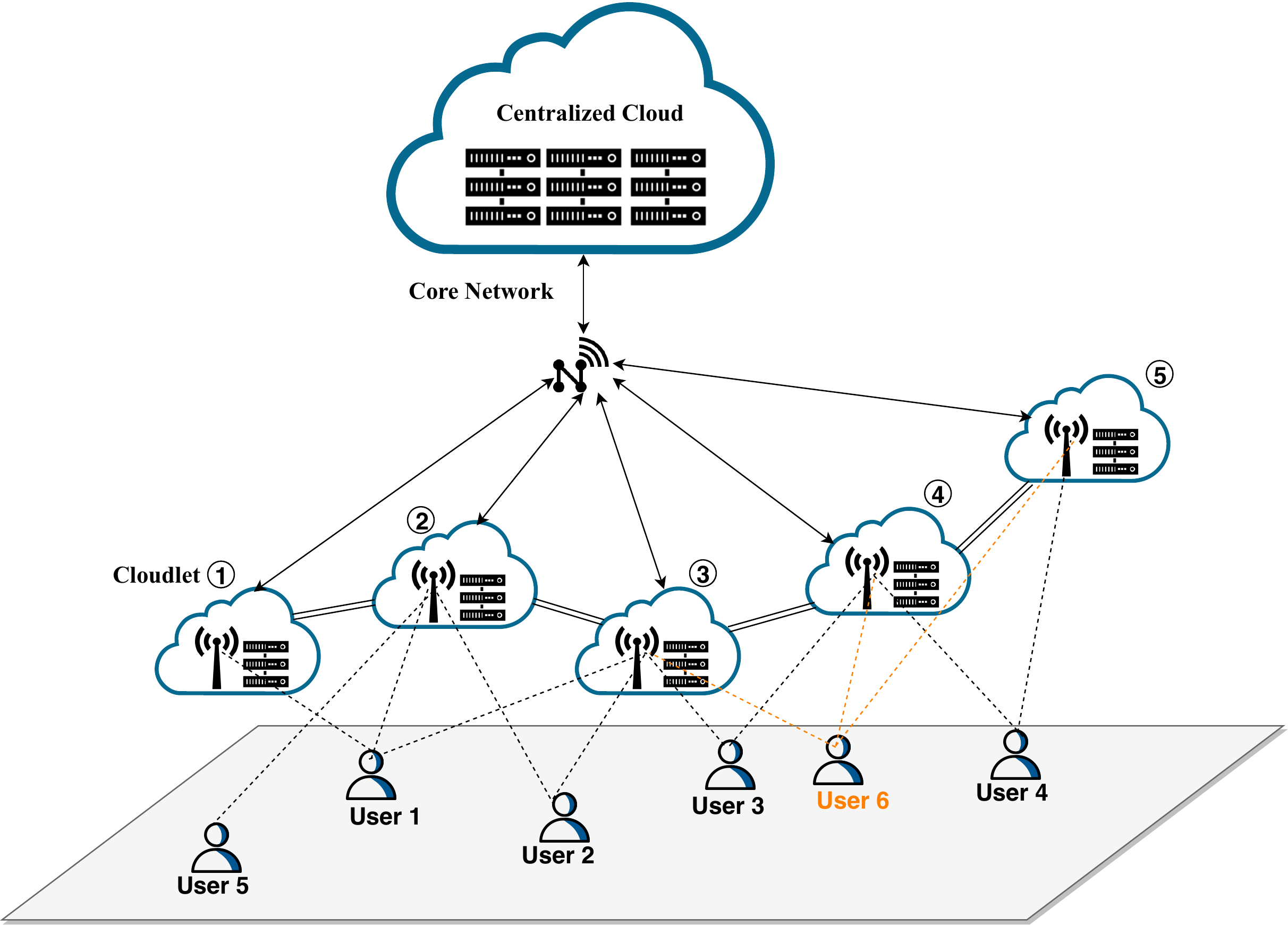}
	\caption{EC system.}
	 \label{fig:framework} 
	\vspace*{-0.4cm}
\end{figure}

A decision pair~$(i,j)$ is  made by the  coordinator for each new joining user~$k$, where $i \in \mathcal{M}$ represents the  AP to connect to and $j \in \mathcal{M}$ denotes the assigned edge server at cloudlet~$j \in~\mathcal{M}$. 
Even though a user is connected to its nearby AP, its allocated edge server can be at any  cloudlet in the EC system. 
 If assigned AP~$i$ and  edge server~$j$ of user~$k$ are not associated with each other (i.e., not in the same cloudlet), the system transfers its task from cloudlet~$i$ to cloudlet~$j$.

The system state is represented by $\mathcal{I}(\tau)=(P(\tau),Q(\tau))$ at any time instant~$\tau$, where~$P(\tau)$
and~$Q(\tau)$ represent the status of the system in terms of 
users connected to all APs and computational tasks  served by all  edge servers at~$\tau$, respectively. 
Specifically, at any time~$\tau$, they present the sets of decision variables defined as follows:
\begin{align} p^k_i(\tau) &=
\begin{cases} 
1 & \quad \text{if user $k $ is connected to AP $i$,}\\
0 & \quad \text{otherwise.}
\end{cases} \\
q^k_j(\tau) &=
\begin{cases} 
1 & \quad \text{if user~$k$ is served by cloudlet~$j$,}\\
0 & \quad \text{otherwise.}
\end{cases} 
\end{align} 
Therefore, at time~$\tau$, 
 the total number of users connected to AP~$i$ is $u_i(\tau)=\sum_{k \in \mathcal{N}}^{}p^k_i(\tau)$, 
 the total number of computational tasks of users served by cloudlet~$j$ is $v_j(\tau)=\sum_{k \in \mathcal{N}}^{}q^k_j(\tau)$, and the total number of computational tasks of users sent to cloudlet~$j$ via AP~$i$ is $x_{(i,j)}(\tau)=\sum_{k \in \mathcal{N}}^{} p^k_i(\tau)q^k_j(\tau)$.
 To make the mathematical formulation a linear convex, we can linearize $x_{(i,j)}(\tau)$. 
 We first define a binary decision variable~$y_{ij}$, and define the following set of constraints:
\begin{equation}\label{eq-y}
y_{ij}(\tau) \geq p^k_i(\tau) + q^k_j(\tau) -1,\; \forall i,j \in \mathcal{M}.
\end{equation}
 to ensure that~$y_{ij}(\tau) $ is one if both~$p^k_i(\tau)$ and~$q^k_j(\tau)$ are one; and zero otherwise.  
 We then define:
\begin{equation} \label{eq-x}
  x_{(i,j)}(\tau)=\sum_{i \in \mathcal{M}} \sum_{j \in \mathcal{M}}  y_{ij}(\tau).
\end{equation}  

When a new user joins the system, it will impact the system and  all existing users. Users using the same AP, edge server, or both as the new user may experience an additional delay. 
We model the new system state~$\hat{\mathcal{I}}$ after a decision pair~$(i^*,j^*)$ is assigned to a new joining user~$k$ at  time~$\tau$, assuming existing users in the system are following their assigned decision pairs. 
We  simply simulate the new system state~$\hat{\mathcal{I}}=(\hat{P},\hat{Q})$  considering~$\hat{p}^{k}_{i^*}(\tau)=1$ and $\hat{q}^{k}_{j^*}(\tau)=1$ while the states of other APs $i \neq i^*$ and edge servers $j \neq j^*$ remain unchanged.
In addition,  $\hat{u}_{i^*}(\tau)=u_{i^*}(\tau-1)+1$ and $\hat{v}_{j^*}(\tau)=v_{j^*}(\tau-1)+1$. 

When a user leaves the system, 
the coordinator updates the system state by releasing the communication and computing resources allocated to  that user. Specifically, considering an  assigned decision pair~$(i^*, j^*)$, we have $\hat{u}_{i^*}(\tau)=u_{i^*}(\tau-1)-1$, $\hat{v}_{j^*}(\tau)=v_{j^*}(\tau-1)-1$, and other related parameters will be updated, accordingly. 

We next describe the AP  allocation model, the service placement model, and the end-to-end latency model in detail.

\subsection{Access Point Allocation Model}
 As mentioned, each cloudlet is associated with an AP. 
 It is possible that a user is within a range of multiple APs and can access any of them, but the user will be connected to only one AP for each of its tasks.  
 We define an indicator variable~$\delta_{ki}$ that characterizes the availability of AP~$i \in \mathcal{M}$ to user~$k \in \mathcal{N}$ at  time~$\tau$ as follows: 
 \begin{equation*}\label{AP_avail}
		\delta_{ki}(\tau) = \left\{\begin{array}{l}1 \;\;\;\;\; \text{if AP $i$ is available to user $k$ at $\tau$,}\\ 0 \;\;\;\;\; \text{otherwise}. 
		\end{array}\right.
 \end{equation*}
 This indicates whether  user~$k$ can connect to AP~$i$ at time~$\tau$ or not. Therefore, we have the following constraint for AP selection:
 \begin{equation}\label{AP_selection}
	 \sum_{i \in \mathcal{M}}^{} \delta_{ki}(\tau)p^k_{i}(\tau)=1, \forall k \in \mathcal{N}, 
 \end{equation}
 which implies that each user can only connect to one available AP at time~$\tau$.

 If too many users choose to connect to the same AP simultaneously, they may incur severe interference, which  eventually leads to lower uplink data rate. This would negatively affect the performance of computation offloading in the EC system. 
 Therefore, the system needs to guarantee  the following:
 \begin{equation}\label{p_cons}
u_i(\tau) \leq \mathbf{P}_i,\; \forall i \in \mathcal{M}.
\end{equation}

\subsection{Service Placement Model}
When a  user requests an edge service, the EC coordinator needs to decide where to properly place  computational resources (e.g., VM or Container) to serve this user. Specifically, the requested resources can be hosted on any cloudlet~$j \in \mathcal{M}$ that satisfies the  QoE requirements of the user and improves the efficiency of the EC system.  
If tasks can only be executed on
 edge servers  associated with their connected APs, as
the number of arriving users increases, the EC system will be
 overloaded quickly leading to unsatisfactory performance. Therefore, we seek to find a proper service placement for each user's task.

Each user's task is served by only one cloudlet, thus we  have:
\begin{equation}\label{VM_place2}
\sum_{j \in \mathcal{M}}^{} q^k_j(\tau)=1, \forall k \in \mathcal{N}.
\end{equation}
In addition, the assignment of tasks to edge servers of each cloudlet should not exceed its capacity: 
\begin{equation}\label{capa_size}
	\sum_{k \in \mathcal{N}}^{} q^k_j(\tau)D_k \leq \mathbf{D}_j(\tau), \forall j \in \mathcal{M}.
\end{equation}
Moreover, we need to ensure that the total number of users connecting to APs is exactly equal to the total number of users served by the cloudlets all the time. Therefore, we have:
\begin{equation}\label{in=out}
\sum_{i \in \mathcal{M}}^{} {u}_i(\tau) = 	\sum_{j \in \mathcal{M}}^{} {v}_j(\tau).
\end{equation}

\subsection{End-to-End Latency Model}\label{subsec:delay}
End-to-end latency includes the network delay of transmitting the data to a cloudlet (communication delay), the processing time at the cloudlet (computation delay),   and finally the network transport delay of transmitting the results to the user's device (communication delay).

\vspace*{0.2cm}
\noindent\emph{C.1) Communication Delay.}
 The communication delay consists of the transmission delay of the user connecting to a proper AP and the transferring delay of the AP relaying to a proper edge server if the connected AP and edge server are not associated with each other. 

\vspace*{0.2cm}
\noindent\emph{Transmission Delay.}
Transmission delay is determined by the wireless communication conditions (e.g., the number of users connected to same AP).
Assuming the bandwidth of an AP is equally allocated to all users connecting to it, 
the bandwidth  allocated to  user~$k$ at time~$\tau$ from AP~$i$ is $r_{ki}(\tau)=\mathbf{B}_i/u_i(\tau)$.
Therefore, the uplink transmission delay of offloading task~$k$ to AP~$i$ at time~$\tau$ is calculated  as:
\begin{equation}\label{delay_ac}
\Lambda^{t}_{k,i}(\tau) =  \frac{D_k}{r_{ki}(\tau)}=\frac{D_k u_i(\tau)}{\mathbf{B}_i}.
\end{equation}

\vspace*{0.2cm}
\noindent \emph{Transferring Delay.}
When the connected AP~$i$ is not associated with the assigned edge server~$j$, i.e., $i \neq j$, we consider a transferring delay~$\Lambda^{f}_{i,j}(\tau)$ as a function of hop distance between the cloudlet of the connected AP and the desired cloudlet. This is due to the fact that    the cloudlets are interconnected via LAN and their physical distance is small. 
Obviously, if $i=j$, there is no transferring delay, i.e., $\Lambda^{f}_{i,j}(\tau)=0$.

Similar to many studies (e.g.,~\cite{chen2014decentralized,ma2015game}), 
we neglect the  delay from the edge server to send the computational results back to the user when the connected AP is associated with the assigned edge server (i.e., no transferring delay). Otherwise, we consider the transferring delay as the total backhaul delay. This is because that the size of computation outcome for many applications or computational tasks (e.g., image recognition)  is usually much smaller than the size of input data.

\vspace*{0.2cm}
\noindent\emph{C.2) Computation Delay}.
We consider  the computational capabilities of a cloudlet are fairly divided among its assigned tasks. 
The computation delay of the task of user~$k$ executed on an edge server of cloudlet~$j$ at time~$\tau$ is calculated using:
\begin{equation}\label{delay_cp}
\Lambda^{c}_{k,j}(\tau) = \frac{C_k v_j(\tau)}{\mathbf{F}_j}, \forall j \in \mathcal{M}.
\end{equation}

\vspace*{0.2cm}
\noindent\emph{C.3) Total Delay}. 
The  end-to-end latency (total delay) experienced by user~$k$ with   an assigned decision pair~$(i,j)$ at time~$\tau$ is as follows:
\begin{equation}\label{delay_sys}
 \Lambda^{l}_{k,(i,j)}(\tau) = \Lambda^{t}_{k,i}(\tau)+2\Lambda^{f}_{i,j}+\Lambda^{c}_{k,j}(\tau)
\end{equation}

Furthermore, the system needs to ensure that the total delay experienced by user~$k$ does not exceed its maximum tolerable end-to-end latency, that is: 
\begin{equation}\label{delay_cons}
{\Lambda}^l_{k,(i,j)}(\tau) \leq T_k, \forall k \in \mathcal{N}.
\end{equation}

\section{Mechanism Design-based Offloading}\label{sec:mecha_design}
Users can be modeled as selfish players that can game the system 
 leading to network congestion, imbalance load, and inefficient resource allocation. 
Algorithmic mechanism design provides a suitable 
approach to incentivize players  to cooperate with the system in order to reach desirable
outcomes. 
 The goal of algorithmic mechanism design is to design a system for  such self-interested players, such that their strategies   at equilibrium lead to expected system performance. 
In this section, we propose a
computation offloading and pricing mechanism (COPM) to solve the dynamic
computation offloading problem  in edge computing
based on  algorithmic mechanism design.

\subsection{Utility Functions}\label{pricing}
\noindent\emph{A.1) User-Centric Model}. 
A  user $k \in \mathcal{N}$ sends its offloading request in the form of $(C_k,D_k,T_k)$ at  time~$\tau$ to the EC system. 
The valuation of user~$k$ for a decision pair~$(i,j)$ considering~$\Lambda^{l}_{k,(i,j)}(\tau)$ (experienced latency) and~$T_k$
is defined as:
\begin{equation}\label{value_user}
V_{k,T_k}^{(i,j)}(\tau)=\psi_k(T_k-\Lambda^{l}_{k,(i,j)}(\tau)), 
\end{equation}
where~$\psi_k$ is a constant value representing user~$k$'s monetary preference per unit of time for its QoE.

The \textit{utility of  user~$k$} when it follows assigned decision pair~$(i,j)$ at time~$\tau$ is determined by:
\begin{equation}\label{utility_user}
U_{k,T_k}^{(i,j)}(\tau) = \underbrace{V_{k,T_k}^{(i,j)}(\tau)}_{\text{\small valuation}}
-   \underbrace{w_{k}^{(i,j)}(\tau)}_{\text{\small payment}}, 
\end{equation}
where~$w_{k}^{(i,j)}(\tau)$ is the  payment of the user  for completing its task through the assigned decision pair. We assume that users are risk-neutral and want to maximize their utilities.

\vspace*{0.2cm}
\noindent\emph{A.2) System-Centric Model}.
The EC system aims to maximize the social surplus of all current users (excluding new users) 
while satisfying the QoE of each user.

When new user~$k$ joins  the system with an assigned decision pair~$(i^*,j^*)$ at time~$\tau$, we define the \textit{valuation of the system} as follows: 
\begin{equation}
\begin{aligned}
V_{s,(i^*,j^*)}(\tau) = &  \sum_{n \in \mathcal{N} \setminus k} \sum_{i \in \mathcal{M}} \sum_{j \in \mathcal{M}} A \cdot  \Lambda^{l}_{n,(i,j)}(\tau), 
\end{aligned}
\end{equation}
where~$A$ is a vector of parameters~$\alpha, \beta, \gamma$ representing the monetary preferences per unit of time in transmission,  transferring, and computation parts of the offloading, respectively.
In particular, $\alpha_{i}$ is the monetary value of time for AP~$i \in \mathcal{M}$;  
$\beta_{(i,j)}$ is the monetary value of time for transferring a task to  edge server~$j$ via assigned AP~$i$; and
$\gamma_j$ is the monetary value of time for  edge server~$j \in \mathcal{M}$. 
 Therefore, the \textit{valuation of the system} is calculated as:
\begin{equation}\label{value_system}
\begin{aligned}
V_{s,(i^*,j^*)}(\tau) = 
\Big( &\sum_{i \in \mathcal{M}}^{}\sum_{n \in \mathcal{N} \setminus k}^{} \alpha_{i}\hat{\Lambda}^{t}_{n,i}(\tau)+ \\
& \sum_{i \in \mathcal{M}}^{}\sum_{j \in \mathcal{M}}^{}\beta_{(i,j)}  x_{(i,j)}(\tau) 2 \hat{\Lambda}^f_{i,j} + \\
& \sum_{j \in \mathcal{M}}^{}\sum_{n \in \mathcal{N} \setminus k} \gamma_j \hat{\Lambda}^{c}_{n,j}(\tau)
\Big),
\end{aligned}
\end{equation}
where $\hat{\Lambda}^{t}_{n,i}(\tau)$ and $\hat{\Lambda}^{c}_{n,j}(\tau)$ are the new transmission delay and the new computation delay of the current users ($\forall n \in \mathcal{N} \setminus k$) after user~$k$ joins at time~$\tau$, respectively. 
In particular, for all AP $i \neq i^*$  and edge server~$j \neq j^*$, we have $\hat{\Lambda}^{t}_{n,i}(\tau)=\Lambda^{t}_{n,i}(\tau-1)$ and  $\hat{\Lambda}^{c}_{n,j}(\tau)=\Lambda^{c}_{n,j}(\tau-1)$ for all users.  
For AP~$i^*$ and edge server~$j^*$, the value of $\hat{\Lambda}^{t}_{n,i^*}(\tau)$ and $\hat{\Lambda}^{c}_{n,j^*}(\tau)$ will be calculated according to new $\hat{u}_{i^*}({\tau})= u_{i^*}({\tau-1})+1$ and new $\hat{v}_{j^*}({\tau})= v_{j^*}({\tau-1})+1$, respectively.  
Note that the transferring delay between any two cloudlets will not be affected by the new joining user (i.e., $\hat{\Lambda}^f_{i,j} = {\Lambda}^f_{i,j}$), 
since it  depends on their number of hop distances.

The \emph{utility of the system} is defined as:
\begin{equation}\label{utility_system}
U_{s,(i^*,j^*)}(\tau) = \underbrace{w_s(\tau)}_{\text{\small payment}} - \underbrace{V_{s,(i^*,j^*)}(\tau)}_{\text{\small valuation}}
\end{equation}
Note that the EC system is better off as the total delay of all
current users decreases. 
Moreover, the mechanism is budget balanced, where the exchanged  payments are equal. Meaning that:
$w_s(\tau)=w_{k}^{(i,j)}(\tau).$

	\begin{figure*}[b]
			\vspace*{-0.4cm}
		\begin{equation}\label{obj2}
		\begin{aligned}	
		&-V_{s,(i,j)}(\tau)+V_{k,T_k}^{(i,j)}(\tau) \\
	&
		=  - 
		\underbrace{
			\Big\{   
			\sum_{i \in \mathcal{M}}^{}\sum_{n \in \mathcal{N}  \setminus k} \alpha_{i}\Lambda^{t}_{n,i}(\tau-1)
			+ \sum_{i \in \mathcal{M}}^{}\sum_{j \in \mathcal{M}}^{}2\beta_{ij} x_{ij}(\tau-1)\Lambda^f_{(i,j)}  
			+ 
			\sum_{j \in \mathcal{M}}^{}\sum_{n \in \mathcal{N}  \setminus k}\gamma_j \Lambda^{c}_{n,j}(\tau-1)
			\Big\}}_{\text{\small Term 1}}  	+ 
		\underbrace{\psi_k T_k}_{\text{\small Term 2}}  \\&\;\;\;\;\;  -
		\Bigg\{ 
	\Big\{
			\underbrace{
			\alpha_{i^*}\sum_{n \in \mathcal{N}  \setminus k}(\hat{\Lambda}^t_{n,i^*}(\tau)-\Lambda^t_{n,i^*}(\tau-1))   +
			2\beta_{i^*j^*}x_{i^*j^*}(\tau-1) 	
			(
			\underbrace{ 
				\hat{\Lambda}^f_{(i^*,j^*)} 	 - \Lambda^f_{(i^*,j^*)}}_{\text{\small =0}} 
			)	  +
			\gamma_{j^*} \sum_{n \in \mathcal{N}  \setminus k}(  \hat{\Lambda}^c_{n,j^*}(\tau) - \Lambda^c_{n,j^*}(\tau-1) )
			\Big\}}_{\text{\small the value of increase in  total delay of the current users (excluding the new user) when new user $k$ joins}}
			\\&\;\;\;\;\;   +
		\underbrace{
			\Big\{   \psi_k \big(\hat{\Lambda}^{t}_{k,i^*}(\tau)+  2\hat{\Lambda}^f_{(i^*,j^*)} + \hat{\Lambda}^c_{k,j^*}(\tau)  \big) \Big\}}_{\text{\small the value of total delay of new user $k$}}
		\Bigg\}  \\&
		\end{aligned} 
		\end{equation}
		\vspace*{-0.8cm}
	\end{figure*}

\subsection{COPM: Computation Offloading and Pricing Mechanism} 
The state of the EC system changes dynamically over time by the arrival and departure of  users. 
At any time  a new user requests to join, we define a \emph{game} between the new  user as a player and all current users in the EC system. 
The objective is to maximize the utilities of both the new user and the EC system. 
We propose the \textit{instant social surplus} to handle the dynamic changes of the EC system: 
\begin{defn}[\textbf{Instant social surplus}]\label{def-surplus}
	\textit{
	The instant social surplus at any time~$\tau$ is the sum of the utility of new user 
	and the utility of the system:} 
\begin{equation}\label{final_obj_system}
U_{s,(i^*,j^*)}(\tau) + U_{k,T_k}^{(i,j)}(\tau) = V_{k,T_k}^{(i,j)}(\tau) - V_{s,(i^*,j^*)}(\tau). 
\end{equation}
\end{defn}

The EC system aims to assign user~$k$ a proper decision pair~$(i,j)$ 
 maximizing the instant social surplus while guaranteeing  this user's QoE and the system's capacity constraints.

Users can choose to lie about their true preferences (i.e., maximum tolerable end-to-end latency) in order to increase their own utility. Such an action  could inversely decrease the overall system efficiency. Therefore, designing an \textit{incentive-compatible} mechanism in which  users have no incentive to lie about their true preferences is extremely crucial in reality. In an incentive-compatible mechanism, truth-telling is a dominant strategy. As a result, it never pays off for any user  to deviate from reporting its true preference,
irrespective of what the other users 
report as their preferences.

We  propose an optimal  incentive-compatible offloading mechanism, COPM, that consists of a decision pair allocation scheme and a payment determination scheme. 
To achieve  incentive compatibility, we need to design an optimal decision pair allocation scheme (subsection B.1) along with a payment function (subsection B.2) designed based on Vickrey-Clarke-Groves (VCG) pricing~\cite{nisan2007algorithmic}. 
We  describe our offloading mechanism design in detail in the following.

\vspace*{0.2cm}
\noindent\emph{B.1) Decision Pair Allocation Scheme}. 
The goal of the EC system is to allocate an optimal decision pair to each new joining user in order to maximize the instant social surplus while satisfying the user's preference. 
We define the Maximization of Instant Social Surplus problem, called  MISS, as follows: 
\begin{equation}\label{obj1}
\max_{(i,j) \in \mathcal{P}} -V_{s,(i,j)}(\tau)+V_{k,T_k}^{(i,j)}(\tau)
\end{equation}
\begin{align*}
s.t.\;\;\;\; (\ref{eq-y})(\ref{eq-x})(\ref{AP_selection})(\ref{p_cons})(\ref{VM_place2})(\ref{capa_size})(\ref{in=out})(\ref{delay_cons}) {\text{ and}}\\
{\text{integrality for the decision variables}},
\end{align*}
where $\mathcal{P}$ represents the set of all feasible decision pairs for user $k$.

Since the  MISS optimization problem is  complicated, we  use  the factorization techniques to obtain a simplified version of MISS, called MISS$^2$. 
\begin{observation} 
	\textit{
	The MISS problem is equivalent to finding a decision pair that minimizes the sum of increase in the total delay of all current users (not including the new user)  and the total delay of the new user itself. }
\end{observation} 

	According to Eq.~(\ref{value_user}) and~(\ref{value_system}), the objective function of MISS in Eq.~(\ref{obj1}) can be rewritten as Eq.~(\ref{obj2}) (See Appendix for the detailed proof). 

Since the EC system knows the state of the  system, 
it is easy to find that in Eq.~(\ref{obj2}), Term~1 representing the utility of the 
system before new user~$k$ joins (i.e.,~$V_{s,(i,j)}(\tau)$) is constant. Term~2 is also constant as~$\psi_k$ and~$T_k$ do not depend on the decision pair. 
	  Since the values of $\alpha_i$, $\gamma_j$, and $\psi_k$ are predefined constants, 
	  the MISS problem has an equivalent minimization problem defined as MISS$^2$ as follows: 
	\begin{equation}\label{obj3}
	\begin{aligned}
	\min_{(i,j) \in \mathcal{P}} 
	\Big\{&
	\alpha_{i}\underbrace{\sum_{n \in \mathcal{N} \setminus k} (\hat{\Lambda}^{t}_{n,i}(\tau) - \Lambda^{t}_{n,i}(\tau-1)  )}_{\text{\small  increase in  transmission delay}}   + \\
	& \gamma_{j} \underbrace{\sum_{n \in \mathcal{N} \setminus k}(\hat{\Lambda}^c_{n,j}(\tau) - \Lambda^c_{n,j}(\tau-1) )}_{\text{\small  increase in computation delay}}
	+ \\
	& \psi_k \underbrace{(\hat{\Lambda}^{t}_{k,i}(\tau)+  2\hat{\Lambda}^f_{(i,j)} + \hat{\Lambda}^c_{k,j}(\tau))}_{\text{\small   total delay of new user}} 
	\Big\}
	\end{aligned}
	\end{equation} \vspace*{-0.4cm}
	\begin{align*}
	s.t.\;\;\;\; (\ref{eq-y})(\ref{eq-x})(\ref{AP_selection})(\ref{p_cons})(\ref{VM_place2})(\ref{capa_size})(\ref{in=out})(\ref{delay_cons}) {\text{ and}}\\
{\text{integrality for the decision variables}}.
	\end{align*} 
	Therefore, the objective 
	now becomes to find an optimal decision pair~$(i^*,j^*)$ for 
	user~$k$ such that the sum of  the increase in   total delay of all current users (not including   user~$k$)  after   user~$k$ joins (first two terms of Eq.~(\ref{obj3}) in MISS$^2$) and the  total delay of   user~$k$ is minimized (third term of MISS$^2$). 
After COPM calculates a proper decision pair for each new user by solving MISS$^2$, 
it calculates a corresponding payment for each user.

\vspace*{0.2cm}
\noindent\emph{B.2) Payment Determination Scheme}. 
After solving the MISS$^2$ problem, an optimal decision pair~$(i^*,j^*)$ is calculated by the EC coordinator for  each new joined user~$k$. The coordinator then needs to compute their payments (e.g., each user~$k$ should pay for using its assigned AP~$i^*$ and edge server~$j^*$). We define the payment based on the \textit{marginal cost pricing} as follows:
\begin{equation}\label{pricing2}
\begin{aligned}
w_{k}^{(i^*,j^*)}&(\tau)=V_{s,(i^*,j^*)}(\tau) - V_{s,(i,j)}(\tau-1)   = \\
   \big\{ &\alpha_{i^*}\sum_{n \in \mathcal{N} \setminus k} \big( 	\hat{\Lambda}^{t}_{n,i^*}(\tau) - \Lambda^{t}_{n,i^*}(\tau-1)  \big) + \\
& \gamma_{j^*} \sum_{n \in \mathcal{N} \setminus k} \big(  \hat{\Lambda}^c_{n,j^*}(\tau) - \Lambda^c_{n,j^*}(\tau-1) \big)
\big\},
\end{aligned}
\end{equation}
where~$V_{s,(i,j)}(\tau-1)$ represents the valuation of the EC system right before user~$k$ joins, and~$V_{s,(i^*,j^*)}(\tau)$ denotes the valuation of the EC system of all current users (not including the  joined  user~$k$), 
calculated according to Eq.~(\ref{value_system}). 
Considering the first two terms of Eq.~(\ref{obj3}), the  payment of each user~$k$ is exactly equal to the increase in  valuations of other current users in the EC system. 

Our proposed mechanism, COPM, is incentive compatible.
To prove this, we firstly introduce the definition of weakly dominant strategy in our mechanism in the following.

\begin{defn}[\textbf{Weakly dominant strategy}]\label{def-dominant}
	\textit{A declared maximum tolerable end-to-end latency of each new joined user is a weakly dominant strategy if and only if 
		it provides at least the same utility  for all the other latency values of this user,   regardless of what  other users in EC do.
	}
\end{defn}
Any mechanism is  incentive compatible if it is a weakly-dominant strategy for users to reveal their private information (declare true latency). 

We now prove the incentive compatibility of our COPM mechanism whenever a new user joins  the EC system. 
\begin{theorem}
	\textit{Given an assigned decision pair~$(i^*,j^*)$ by the decision pair allocation scheme and an assigned  payment~$w_{k}^{(i^*,j^*)}(\tau)$ by the payment determination scheme, declaring the true maximum tolerable end-to-end latency~$T_k$ is a weakly dominant strategy of a new user~$k$ in our COPM mechanism.}
\end{theorem}
\begin{proof}
	It is clear that  declaring a latency~$T_k^{\prime}$ different from the true~$T_k$ may change the optimal decision pair of user~$k$. 
	Every new joined user would like to receive a decision pair that gives  the maximum utility value, and it may choose to misreport to increase its utility.

	
	We claim that user~$k$ maximizes its own utility by declaring its true maximum tolerable end-to-end latency~$T_k$, i.e.,~$U_{k,T_k}^{(i^*,j^*)}(\tau) \geq U_{k,T_k^{\prime}}^{(i^\prime,j^\prime)}(\tau)$, where $(i^\prime, j^\prime)$ is the new decision pair corresponding to any other declared~$T_k^{\prime}$ different from true~$T_k$. The proof is by contradiction.
	We assume that user $k$ maximizes its own utility by declaring $T_k^{\prime}\neq T_k$, which means:
	\begin{equation}\label{proof2}
		U_{k,T_k^{\prime}}^{(i^\prime,j^\prime)}(\tau) > U_{k,T_k}^{(i^*,j^*)}(\tau)
	\end{equation}	
	\begin{align*}
		&s.t.\;\;\;\;
		\hat{\Lambda}^s_{k,(i^\prime,j^\prime)}(\tau) \leq T_k, 
	\end{align*}
	where the constraint shows that the new decision pair~$(i^\prime,j^\prime)$ calculated by declared~$T_k^{\prime}$ should be feasible to  user~$k$, which implies that the new decision pair~$(i^\prime,j^\prime)$ is also a feasible solution to the MISS problem (\ref{obj1}).
	
	Based on 	Eq.~(\ref{utility_user}) and~(\ref{pricing2}), we have:
	\begin{align}
		U_{k,T_k}^{(i^*,j^*)}(\tau) = {V_{k,T_k}^{(i^*,j^*)}(\tau)}-   {w_{k}^{(i^*,j^*)}(\tau)} =\\
		V_{k,T_k}^{(i^*,j^*)}(\tau) - (V_{s,(i^*,j^*)}(\tau) - V_{s,(i,j)}(\tau-1)) \nonumber
	\end{align}
	and similarly for~$U_{k,T_k^{\prime}}^{(i^\prime,j^\prime)}(\tau)$.  
	Therefore, we modify  inequality~(\ref{proof2}) as:
	\begin{align} \nonumber
		V_{k,T_k}^{(i^\prime,j^\prime)}&(\tau) - V_{s,(i^\prime,j^\prime)}(\tau) + V_{s,(i,j)}(\tau-1) 
		> \\
		&V_{k,T_k}^{(i^*,j^*)}(\tau) - V_{s,(i^*,j^*)}(\tau) + V_{s,(i,j)}(\tau-1). \nonumber
	\end{align}
	Since user~$k$ has no control over the term $V_{s,(i,j)}(\tau-1)$ (valuation of the EC system right before user~$k$ joins), we  subtract it from both sides of the inequality and hence get:
	\begin{align}\label{proof1}
		V_{k,T_k}^{(i^\prime,j^\prime)}&(\tau) - V_{s,(i^\prime,j^\prime)}(\tau) > \\
		& V_{k,T_k}^{(i^*,j^*)}(\tau) - V_{s,(i^*,j^*)}(\tau). \nonumber
	\end{align}
	
	In contrast, we know that $(i^*,j^*)$ calculated by the decision pair allocation scheme maximizes the MISS problem~(\ref{obj1}) for this user's true maximum tolerable end-to-end latency~$T_k$ while  $(i^\prime,j^\prime)$ is a feasible solution of the MISS problem~(\ref{obj1}). Thus, we have:
	\begin{align}\label{proof0}
		- V_{s,(i^*,j^*)}&(\tau)+V_{k,T_k}^{(i^*,j^*)}(\tau)  \geq \\
		&- V_{s,(i^\prime,j^\prime)}(\tau)+V_{k,T_k}^{(i^\prime,j^\prime)}(\tau). \nonumber 
	\end{align}

	Obviously, the assumed inequality~(\ref{proof1}) contradicts  inequality~(\ref{proof0}). Therefore,~$T_k$ is a weakly dominant strategy, and our COPM mechanism is incentive compatible. 
\end{proof}

\section{DAPA: Online Algorithmic-based Offloading}\label{sec:algorithm}

We now describe our online algorithmic solution for our mechanism by proposing Dynamic Allocation and Pricing Algorithm (DAPA), presented in Algorithm~\ref{alg}.
When any new user~$k$ requests edge service with~$(C_k, D_k, T_k)$ at time~$\tau$, DAPA first finds all available APs and available edge servers for this  user (lines~3-4). 
    Also, DAPA has the information of all current users whose tasks are not yet completed when the new user joins (this is updated based on users leaving the system). This information is in~$\mathcal{S}({\tau}-1)$, and for each existing user~$n$  it consists of allocated AP~$n.i$, allocated edge server~$n.j$, start time~$n.st$, end time~$n.et$, complete time, 
    transmission delay, transferring delay, and computation delay.  
If the arrival time of  user~$k$ is larger than the end time  of any existing user in~$\mathcal{S}({\tau-1})$,  it indicates that these users have completed their tasks and left the system before user~$k$ joins. Thus, DAPA applies $\textsc{Update()}$ function to update the  system state at~$\tau$ by updating $\hat{u}_{i}(\tau) \gets u_{i}(\tau-1)-1$, $\hat{v}_{j}(\tau) \gets v_{j}(\tau-1)-1$,  $\mathbf{\hat{D}}_{j}(\tau) \gets \mathbf{D}_{j}(\tau-1)+D_{n}$, and other related parameters. The information of the completed user~$n$ will be removed from~$\mathcal{S}({\tau})$ (lines~7-10).

	DAPA defines a 2-D array $\mathcal{V}$ and finds  the value of instant social surplus  by calculating Eq.~(\ref{obj3}) for each feasible decision pair (lines~{11-13}). Note that here, DAPA is not solving the MISS or MISS$^2$ problem to find the optimal decision pair, but simply calculating the value of 	Eq.~(\ref{obj3})  having a decision pair~$(i,j)$.  
	   The optimal decision pair~$(i^*,j^*)$ with the minimum value is obtained from~$\mathcal{V}$ ({line~14}).

	   	DAPA uses the $\textsc{ComputeNewDelay()}$ function to check if the reported  maximum tolerable end-to-end latency  of user~$k$ can be met.  Specifically, the calculated $(i^*,j^*)$ is temporarily assigned to user~$k$ and then its total delay~$\Lambda^{l}_{k,(i,j)}(\tau)$ is computed. If~$\Lambda^{l}_{k,(i^*,j^*)}(\tau) \leq T_k$, it implies assigning~$(i^*,j^*)$ to  user~$k$ is feasible (lines~{15-18}) and the corresponding price for using this pair is calculated using Eq.~(\ref{pricing2}). Otherwise, the request of user~$k$ cannot be served  by the EC system, and it will be forwarded to the cloud ({line~20}).

	\setlength{\textfloatsep}{0.4cm}

\begin{algorithm}[t!]
	\caption{DAPA: Dynamic Allocation and Pricing Algorithm for Offloading}\label{alg}  
	\begin{algorithmic}[1]
			\STATE \textbf{{\color{black}Input}:} User~$k$ edge service request: $C_k,D_k,T_k$
			\STATE \textbf{{\color{black}Input}:} System state $\mathcal{I}(\tau-1)=(P(\tau-1),Q(\tau-1))$ 
			\STATE $H^{a} \gets \text{feasible APs for  user}\; k$
			\STATE $H^{e} \gets \text{feasible edge servers for  user}\; k$
			\STATE $\mathcal{S}({\tau-1}) \gets$ {\color{black}the set of information of current users}
			\STATE $\mathcal{V} \gets \emptyset$\quad /*2D array of instant social surplus values~(\ref{obj3})*/
			\FOR {each current user $n \in \mathcal{S}({\tau-1})$}	
				\IF {$k.st > n.et$}
					\STATE $\hat{I}(\tau) \gets \textsc{Update()}$
					\STATE $\mathcal{S}(\tau) \gets \mathcal{S}({\tau-1}) \setminus n$
				\ENDIF
			\ENDFOR
			\FOR {each  AP $i \in H^{a}$}
				\FOR {each  edge server $j \in H^{e}$}
					\STATE $\mathcal{V}[i][j] \gets$ value of Eq.~(\ref{obj3})
				\ENDFOR
			\ENDFOR
			\STATE $(i^*,j^*) \gets \arg\min (\mathcal{V})$ 
			\STATE  $\Lambda^{l}_{k,(i^*,j^*)}(\tau) \gets \textsc{ComputeNewDelay}(k,(i^*,j^*))$ based on Eq.~(\ref{delay_sys})
			\IF {$\Lambda^{l}_{k,(i^*,j^*)}(\tau) \leq T_k$}
				\STATE $w^* \gets$ payment for using $(i^*,j^*)$ based on Eq.~(\ref{pricing2})
				\STATE \textbf{return} $(i^*,j^*), w^*$
			\ELSE
				\STATE user $k$'s request is sent to the cloud
			\ENDIF			
	\end{algorithmic}
\end{algorithm}

\section{Experimental Results}\label{sec:results}

\subsection{Experimental Setup}\label{setup}
\subsubsection{EC System Data}\label{EC_data}
The simulated area is a $500 \times 500 \; \text{m}^2$ square  covered by~$8$ cloudlets,  deployed evenly in this area.  
The effective radius~$r_i$ of coverage of each AP~$i$ is randomly selected from $[75, 100, 125]$ meters in order to generate the values of the indicator variable. 
The coverage areas of cloudlets can  overlap, which indicates each arriving user may have multiple APs to connect to based on its  coordinates. 
We set the bandwidth of APs obeys Gaussian distribution with mean~$\mu=100$~Mbps and standard deviation~$\sigma=0.25\mu$.
The maximum number of users to be served simultaneously by AP~$i$ ($\mathbf{P}_i$) is
uniformly selected from~$[10,30]$.
 The edge servers are heterogeneous, and each edge server can be equipped with multiple  CPU cores. 
The computation capability of edge servers ($\mathbf{F}_j$) is uniformly selected from~$[5,10]$~GHz. 
The transferring delay between two cloudlets is uniformly distributed in~$[0.1,0.5]$ sec. 
The memory capacity~$\mathbf{D}_j$ of each edge server~$j$ is $8$~GB.

\subsubsection{User Data}\label{user_data}
The Poisson process plays an important role in modeling systems, 
as it is usually used in scenarios where the goal is to count the occurrence of certain events 
happenning  at a certain rate but completely at random~\cite{bertsekas2002introduction}. In this paper, we assume that  user arrival events can be modeled as a Poisson process with rate~$\lambda = n_{a}/3600$, where~$n_{a}=1200$ represents the number of users arriving in the EC system within one hour. 
 Each user~$k$ has a computation offloading request, and its location is arbitrary.  
 The data size~$D_k$ of user~$k$ is uniformly selected from~$[5,60]$~MB. 
 To specify the required cycles~$C_k$ of the computational task, we consider the general application type in which~$1$~bit requires~$1000$~cycles to be processed~\cite{kwak2015dream}. 
 We roughly classify the  users' tasks into three categories: 
urgent ($t_k+100~\text{sec}$), mid-urgent ($t_k+200~\text{sec}$), and nonurgent ($t_k+300~\text{sec}$), 
 where~$t_k$ is the minimum total latency for completing the computational task of user~$k$. We assume that the reported~$T_k$ from user~$k$ must be no less than~$t_k$. 
Moreover,~$\psi_k$ is $1$\$/h, 
			$\alpha_i$  is $50$\$/h, 
			and $\gamma_j$ is $50$\$/h. 

\subsection{Performance of Benchmark}\label{benchmark}
We simulate a real-time scenario with a duration of~$3$ hours. 
To evaluate the performance of our proposed mechanism, DAPA, we compare it with two other offloading strategies: 
\begin{enumerate}
		\item \textbf{User Equilibrium (UE):} every  new user selfishly chooses the decision pair with the  minimum total delay.
		\item \textbf{Random Selection (RS):} every  new user randomly chooses a feasible decision pair. 
\end{enumerate}

	We first show the performance of these mechanisms in terms of the workload on APs (Fig.~\ref{fig:result1}) and edge servers (Fig.~\ref{fig:result2}). In particular, these figures show that the dynamics of the number of users~$u_i(\tau)$ connecting to each AP~$i$ and the number of computational tasks~$v_j(\tau)$ on each edge server~$j$ over time. The results show  that DAPA achieves a more efficient allocation to users such that the workload on each AP and each edge server are balanced  overall. Note that RS should be load balanced since
it randomly selects a decision pair, however, the experienced time of users by RS is poor (Fig.~\ref{fig4-3}).

\begin{figure*}[t!]
	\captionsetup{justification=centering}
	\begin{subfigure}[b]{0.29\textwidth}
		\centering
		\includegraphics[height=4.1cm]{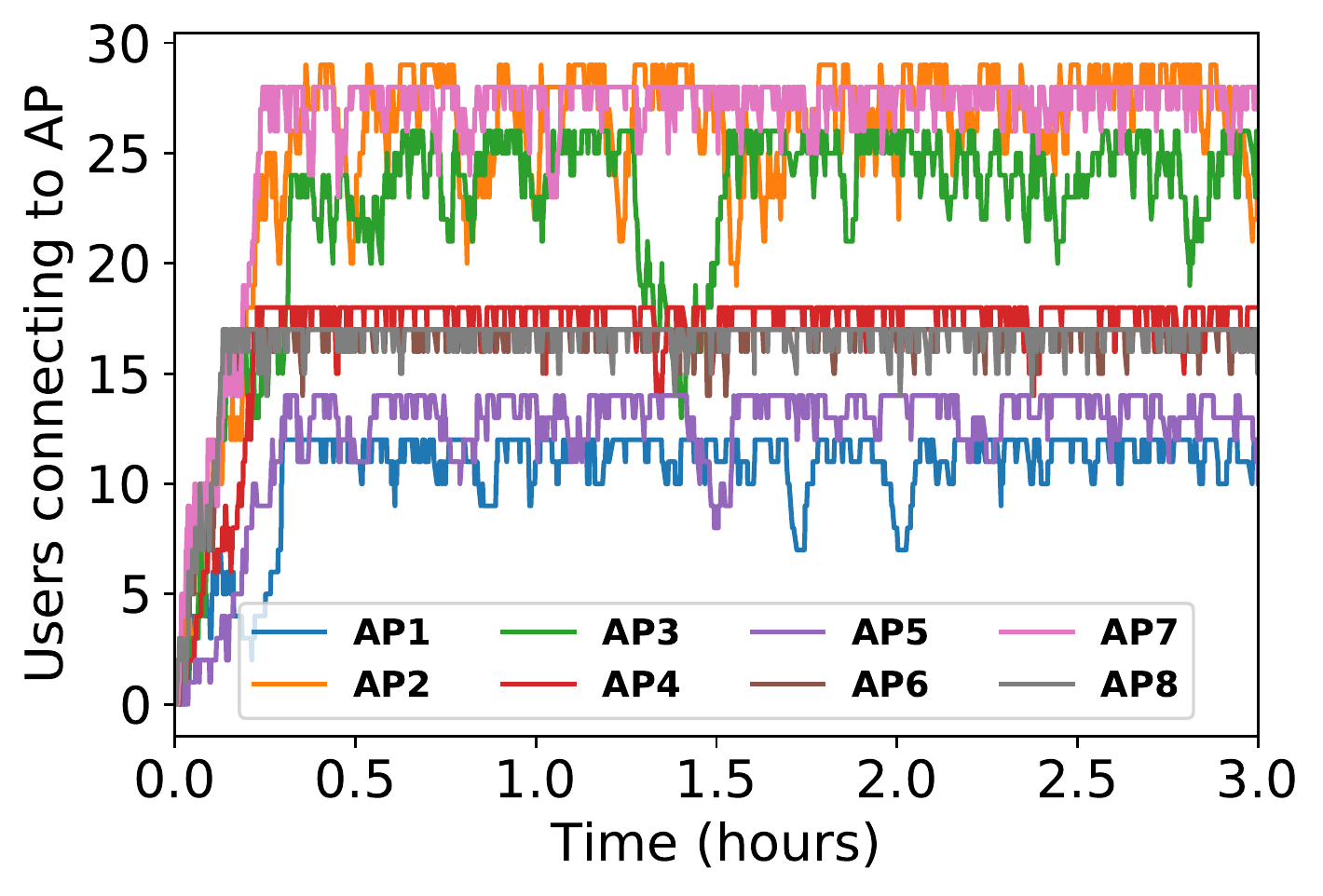}
		\caption[]{{\small User Equilibrium (UE)}} \label{fig1-1} 
	\end{subfigure}
	\qquad
	\begin{subfigure}[b]{0.29\textwidth}
		\centering
		\includegraphics[height=4.1cm]{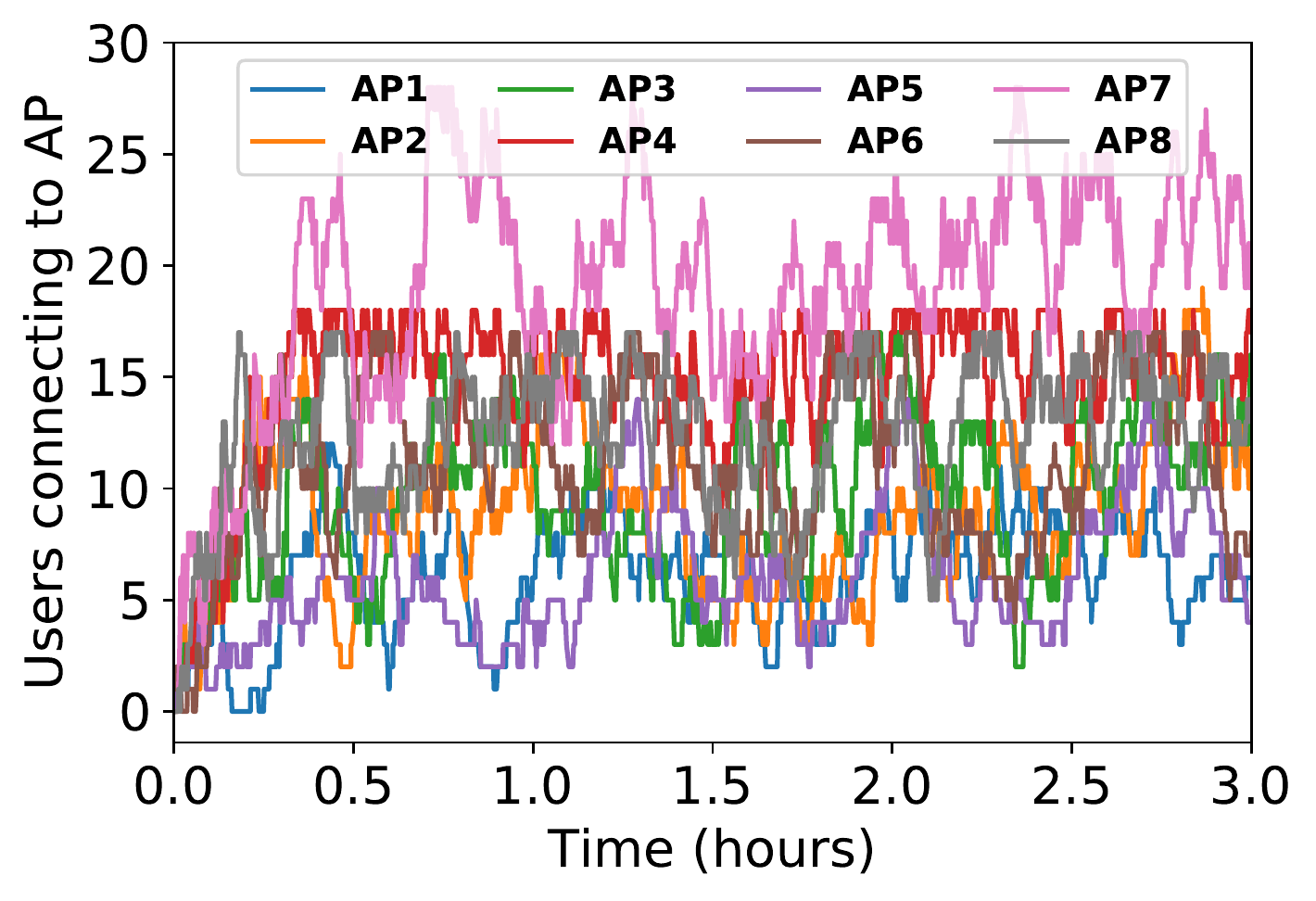}
		\caption[]{{\small  Random Selection (RS)}} \label{fig1-2}
	\end{subfigure}
	\qquad
	\begin{subfigure}[b]{0.29\textwidth}
		\centering
		\includegraphics[height=4.1cm]{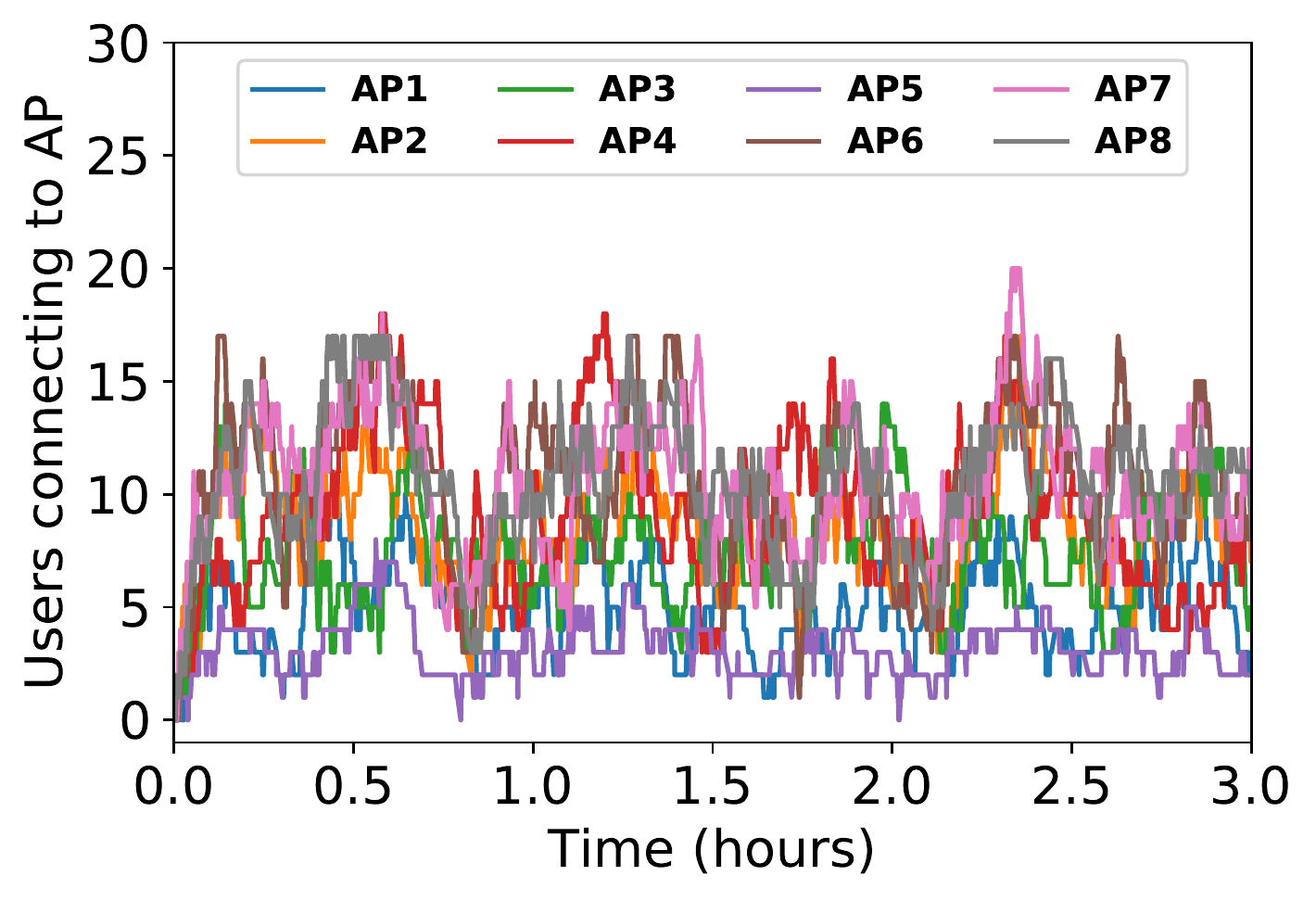}
		\caption[]{{\small DAPA}} \label{fig1-3} 
	\end{subfigure}
	\caption{Analysis of workload on APs}\label{fig:result1}
	\vspace*{-0.2cm}
\end{figure*}

\begin{figure*}[t!]
	\captionsetup{justification=centering}
	\begin{subfigure}[b]{0.29\textwidth}
		\centering
		\includegraphics[height=4.1cm]{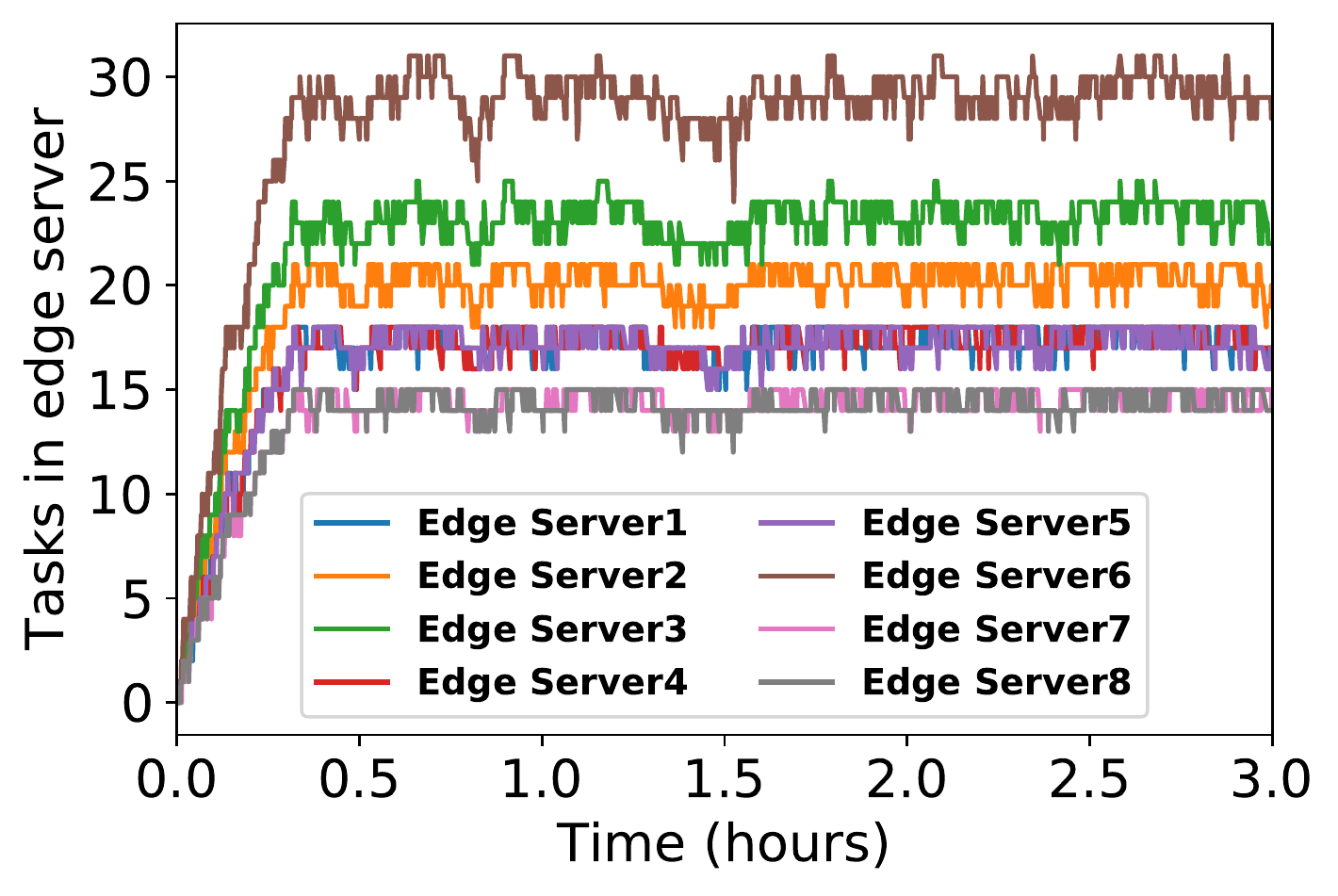}
		\caption[]{{\small  User Equilibrium (UE)}} \label{fig2-1} 
	\end{subfigure}
	\qquad
	\begin{subfigure}[b]{0.29\textwidth}
		\centering
		\includegraphics[height=4.1cm]{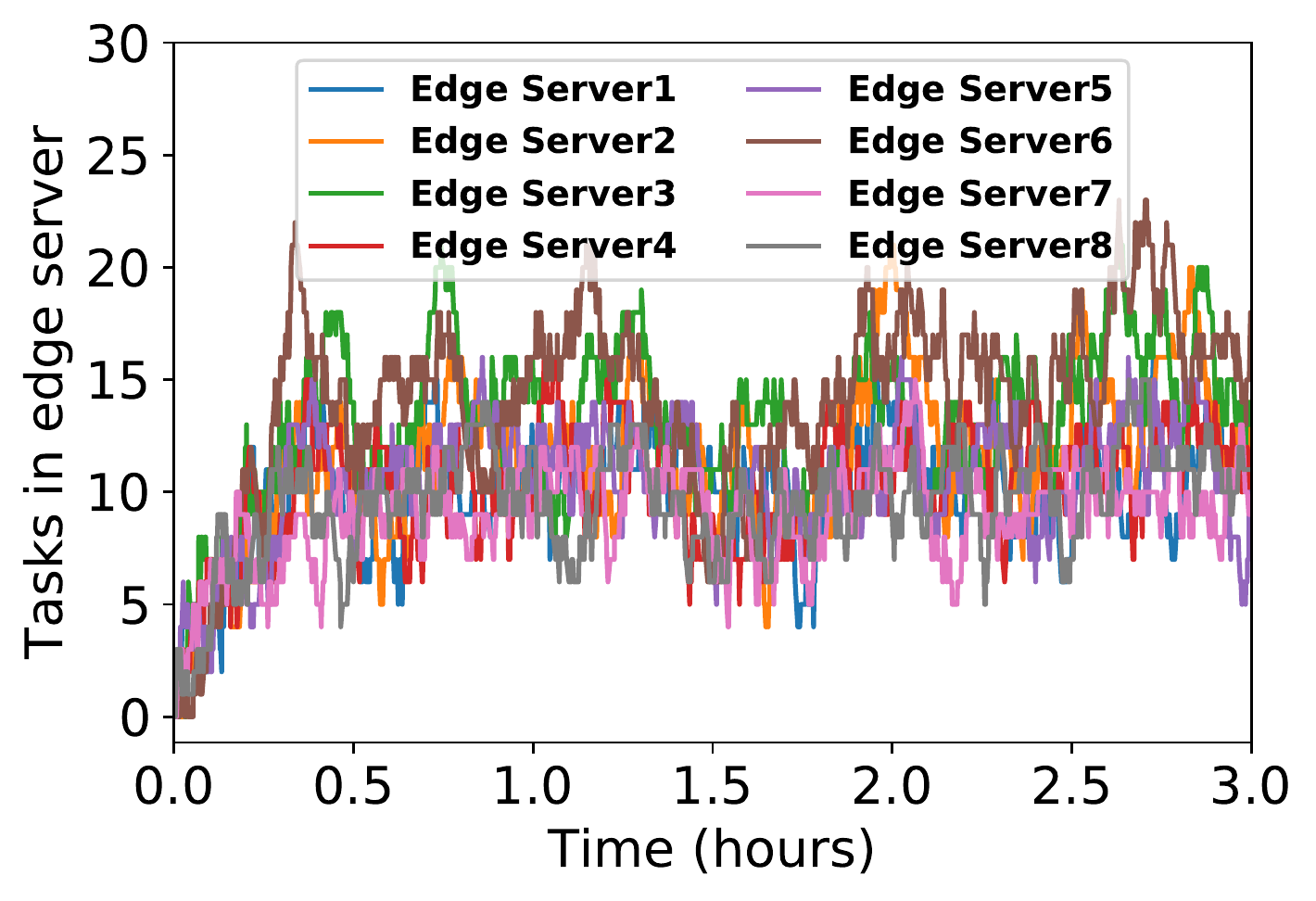}
		\caption[]{{\small  Random Selection (RS)}} \label{fig2-2}
	\end{subfigure}
	\qquad
	\begin{subfigure}[b]{0.29\textwidth}
		\centering
		\includegraphics[height=4.1cm]{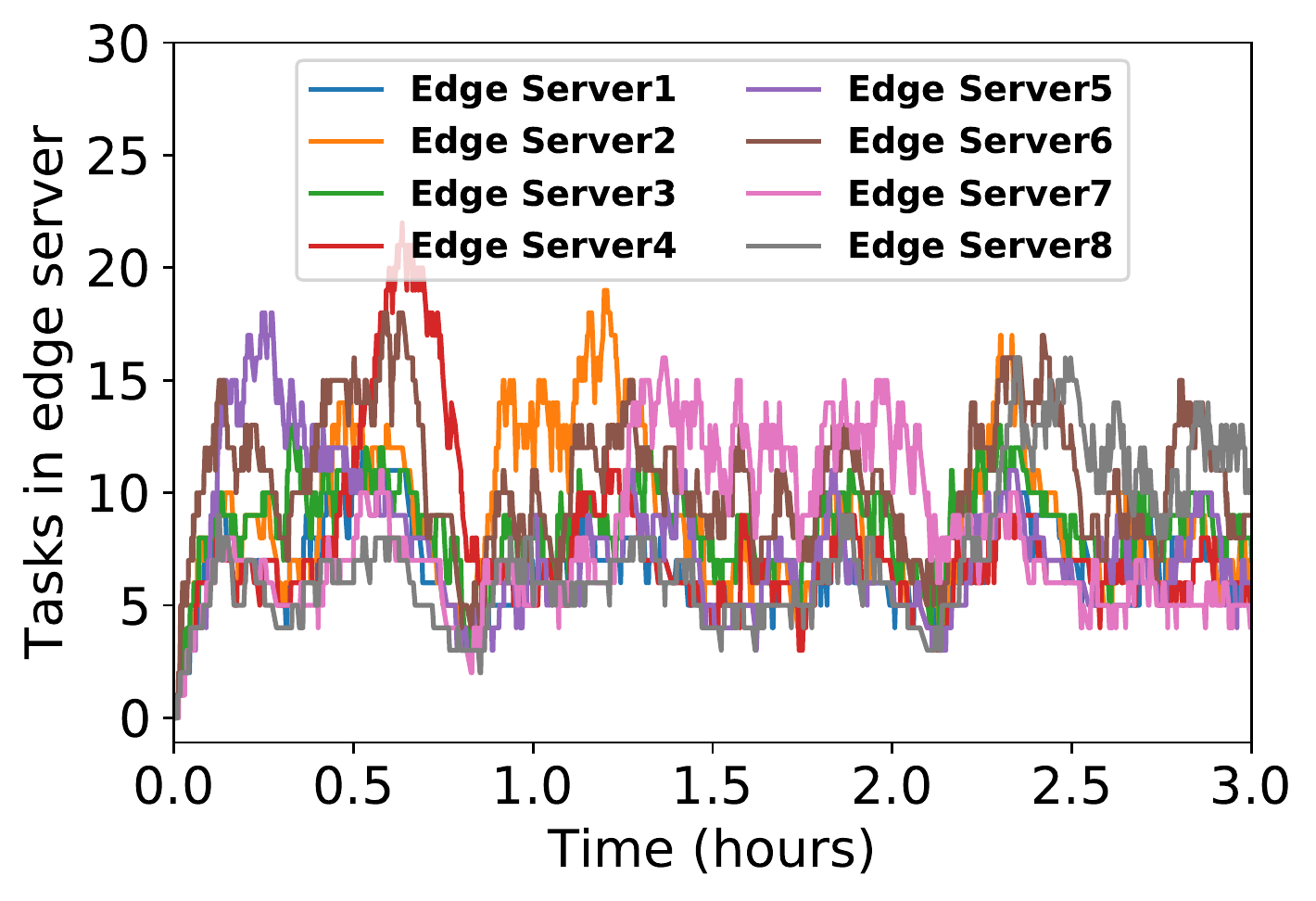}
		\caption[]{{\small DAPA}} \label{fig2-3} 
	\end{subfigure}
	\caption{Analysis of workload on edge servers}\label{fig:result2}
	\vspace*{-0.2cm}
\end{figure*}

	\begin{figure*}[t!]
	\captionsetup{justification=centering}
	\begin{subfigure}[b]{0.29\textwidth}
		\centering
		\includegraphics[height=4.2cm]{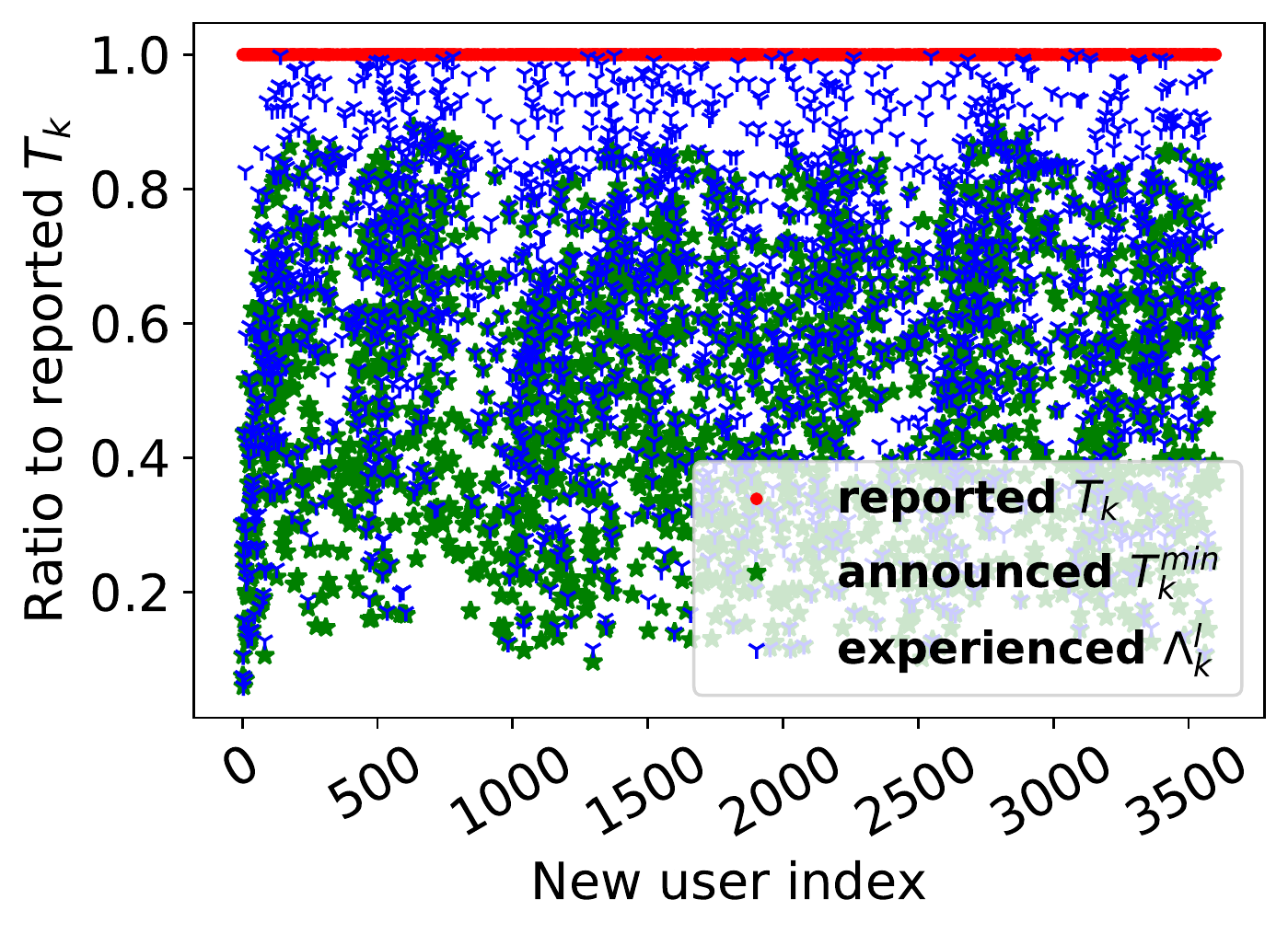}
		\caption[]{{\small  Ratio of latency to the reported~$T_k$} } \label{fig4-2}
	\end{subfigure}
 \qquad
	\begin{subfigure}[b]{0.29\textwidth}
		\centering
		\includegraphics[height=4.2cm]{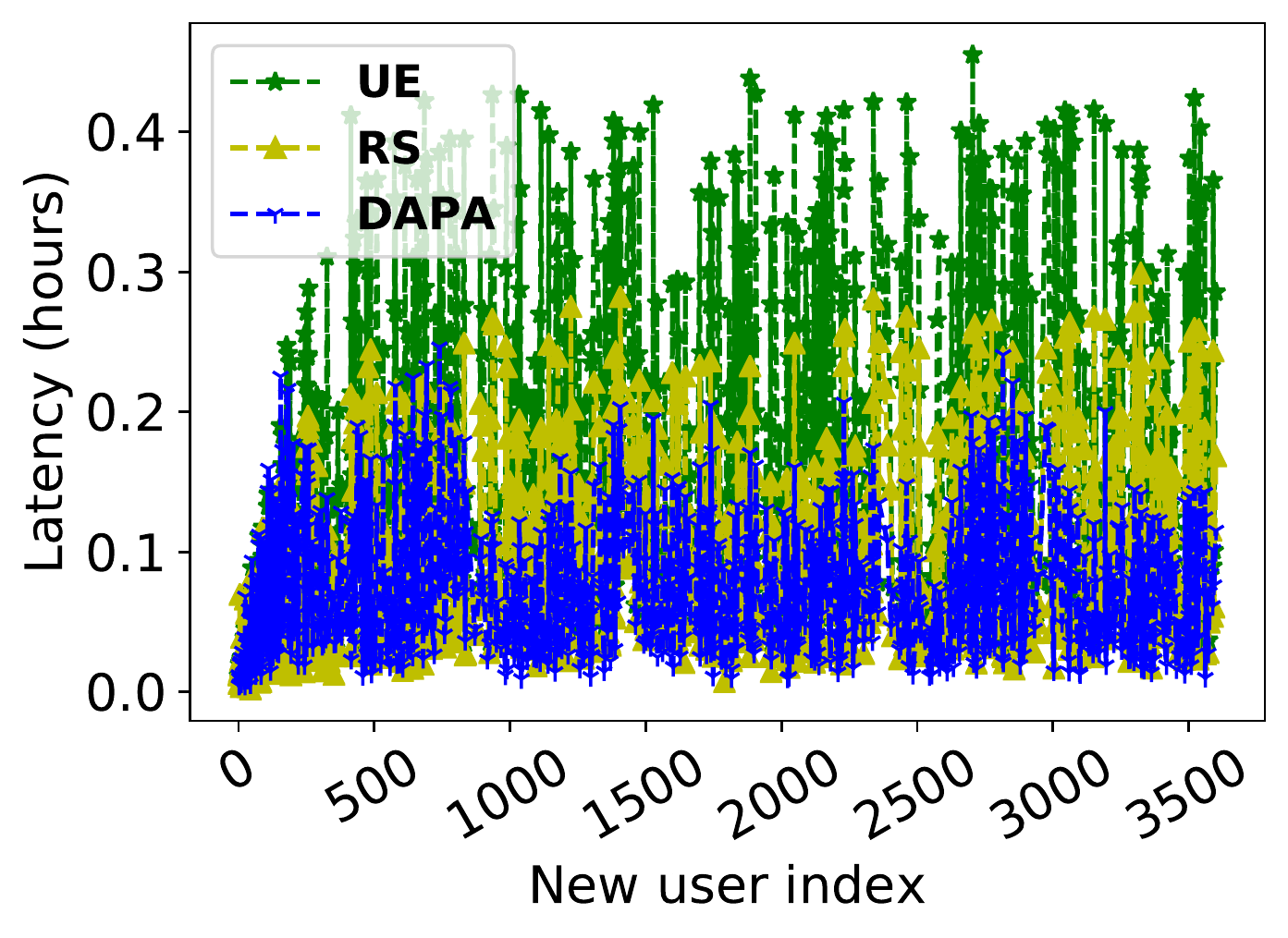}
		\caption[]{{\small End-to-end latency of users }} \label{fig4-3}  
	\end{subfigure}
		\qquad
	\begin{subfigure}[b]{0.29\textwidth}
	\centering
	\includegraphics[height=4.2cm]{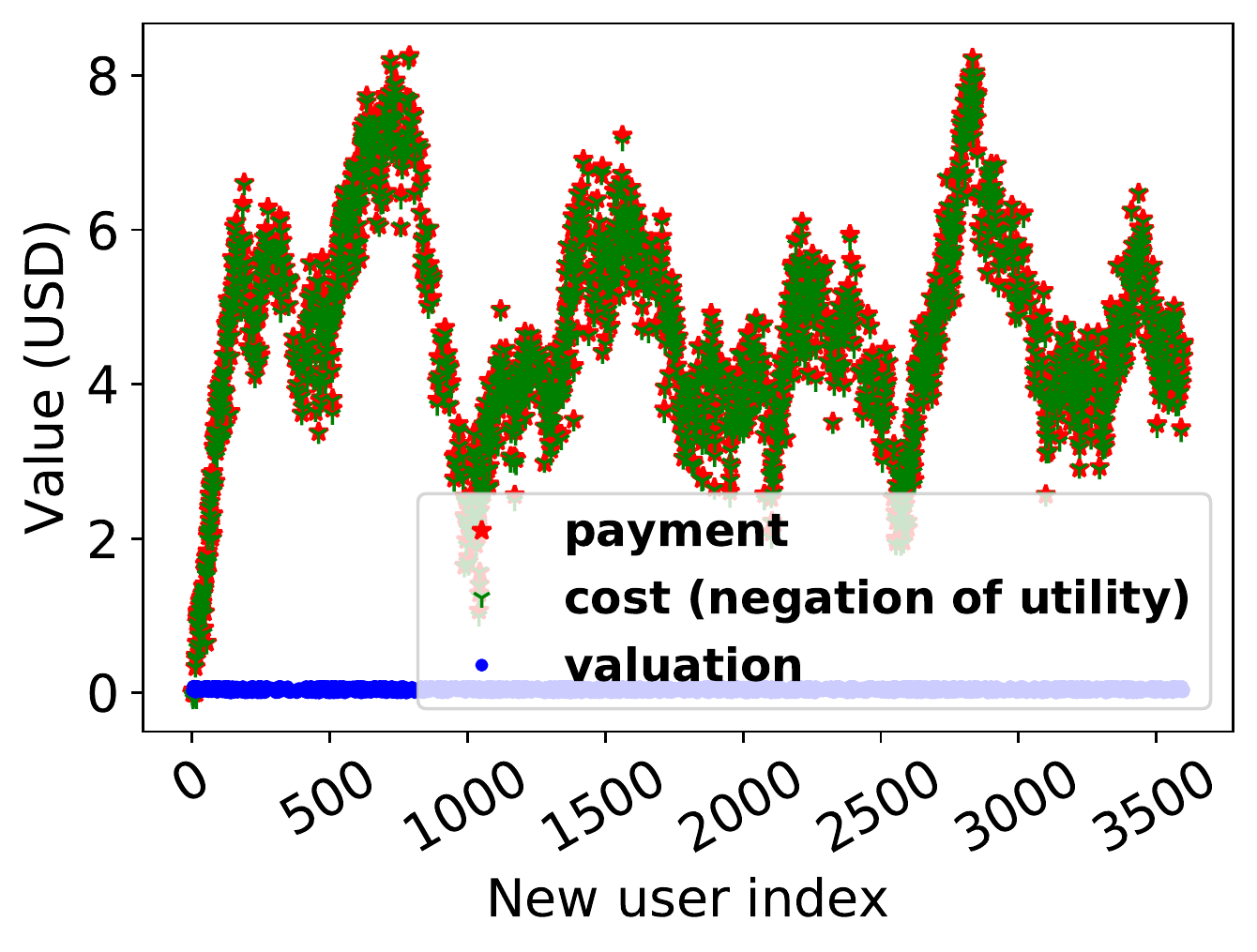}
	\caption[]{{\small  Payment, valuation, and cost of users}} \label{fig4-1} 
	\end{subfigure}
	\caption{Analysis  of pricing}\label{fig:result4} 
		\vspace*{-0.2cm}
\end{figure*} 

	We then investigate the end-to-end latency for completing the computational task of each user when following its assigned decision pair over time. We define~$T_k^{min}$ as the minimum latency that the EC system can provide for completing the computational task of user~$k$, 
that is  equal to the value
of experienced latency by choosing the UE strategy if only this user exists in the system.
	We normalize the experienced latency  of user~$k$ by DAPA ($\Lambda_k^l$) and the minimum latency ($T_k^{min}$) by dividing them by  the reported maximum tolerable latency~$T_k$ of user~$k$. These normalized values are shown in  Fig.~\ref{fig4-2}. 
The results show that the experienced end-to-end latency~$\Lambda_k^l$ of user~$k$ is different but close to the  minimum latency (i.e., $T_k^{min}$ green dots). 
This figure also shows that the users' preferences are satisfied over time since the experienced latency is always less than or equal to the reported maximum  tolerable latency. 

	Moreover, we study the dynamic changes of the experienced end-to-end latency of users over time in Fig.~\ref{fig4-3}. The results show that the proposed DAPA outperforms UE and RS in terms of the  end-to-end latency that users experience for completing their tasks as the number of joined users increases. This is due to the fact  that both  UE and RS do not have  any policy to consider new users' impacts on other existing users in the system. On the contrary, DAPA aims to find the optimal decision pair for each new user with the objective of jointly minimizing 
the sum of the increase in total delay of all
current users (excluding the new user) after the new user joins  and the total delay of the new user. 
UE leads to the worst performance as the number of arrived users increases since it considers selfish assignments and 
the EC system  rapidly becomes overloaded on APs/edge servers (as shown in Figs.~\ref{fig1-1}-\ref{fig2-1}).

	We further evaluate how the EC system makes use of the payments to incentivize each  user to report its own true maximum tolerable end-to-end latency. The payments (red points), valuations, and  costs (i.e., negative utilities) of joined users are shown in Fig.~\ref{fig4-1}. The payment of users who join the system at  the beginning is much less than 
	users who arrive later. This is due to the fact that each user payment depends on the increase in the end-to-end latency of other existing users in the system (according to Eq.~(\ref{pricing2})). 
	When there are fewer users, their payment is lower. 
	For example, Figs.~\ref{fig1-3}  and \ref{fig2-3} show a decrease in the number of users in the system at~0.6-0.8 hour,  that corresponds to about 800th-1000th joining user in Fig.~\ref{fig4-1} with a reduction in their payments. 
	Also, both of the payments and valuations of the users are always non-negative. 
		Additionally, when all users report their maximum tolerable end-to-end latency truthfully, their  costs are minimized (i.e., utilities are maximized) at the equilibrium obtained by DAPA. 

\section{Conclusion}\label{sec:conclusions}
	In this paper, we studied the dynamic computation offloading problem in the EC system. We  formulated the computation offloading optimization problem for   users joining and leaving the system with the objective of jointly optimizing the access point allocation  and service placement problems. To address this challenge, we devised an online incentive-compatible mechanism, DAPA, in which the new  users always declare their true preferences. 
	The effectiveness of the mechanism was  validated by extensive experiments in comparison to User Equilibrium  and Random Selection strategies. 
	For the future work, we plan to consider the effects of user mobility on the computation offloading problem in edge computing.

\vspace*{0.2cm}
\noindent \emph{Acknowledgment.} This research was supported in part by NSF grant CNS-1755913.

{\small
	\newcommand{\BIBdecl}{\setlength{\itemsep}{0.2 em}}
	\bibliographystyle{IEEEtran}
	\bibliography{reference}
}

\appendix
\begin{equation}\label{obj2--}
	\begin{aligned}
		&
		-V_{s,(i,j)}(\tau)+V_{k,T_k}^{(i,j)}(\tau) \\
		&	=-\Big\{\sum_{i \in \mathcal{M}}^{}\sum_{n \in \mathcal{N}  \setminus k} \alpha_{i}\hat{\Lambda}^{t}_{n,i}(\tau)+
		\sum_{i \in \mathcal{M}}^{}\sum_{j \in \mathcal{M}}^{}2\beta_{ij}x_{ij}(\tau-1) \hat{\Lambda}^f_{(i,j)}\\
		&\;\;\;\;\;+
		\sum_{j \in \mathcal{M}}^{}\sum_{n \in \mathcal{N}  \setminus k} \gamma_j \hat{\Lambda}^{c}_{n,j}(\tau) \Big\} + 
		\psi_k \Big(T_k-\Lambda^{l}_{k,(i,j)}(\tau) \Big) \\
		&	=-\Big\{\Big( \sum_{i \in \mathcal{M} \setminus i^*}^{}\sum_{n \in \mathcal{N}  \setminus k} \alpha_{i}\Lambda^{t}_{n,i}(\tau-1)+ \sum_{n \in \mathcal{N}  \setminus k} \alpha_{i^*}\hat{\Lambda}^{t}_{n,i^*}(\tau) \Big)+\\
		&\;\;\;\;\;2\Big(\sum_{i \in \mathcal{M} \setminus i^*}^{}\sum_{j \in \mathcal{M} \setminus j^*}^{}\beta_{ij}x_{ij}(\tau-1) \Lambda^f_{(i,j)} + \beta_{i^*j^*} x_{i^*j^*}(\tau-1)\hat{\Lambda}^f_{(i^*,j^*)}  \Big) \\	
		&\;\;\;\;\;  +\Big(\sum_{j \in \mathcal{M} \setminus j^*}^{}\sum_{n \in \mathcal{N}  \setminus k} \gamma_j \Lambda^{c}_{n,j}(\tau-1)+   \sum_{n \in \mathcal{N}  \setminus k} \gamma_{j^*} \hat{\Lambda}^{c}_{n,j^*}(\tau)  \Big) 	\Big\} \\
		&\;\;\;\;\; +\psi_k \Big( T_k - (\hat{\Lambda}^t_{k,i^*}(\tau) + 2\hat{\Lambda}^f_{(i^*,j^*)} +\hat{\Lambda}^c_{k,j^*}(\tau) )  \Big) \\
		&	 =- \Big\{\sum_{i \in \mathcal{M} \setminus i^*}^{}\sum_{n \in \mathcal{N}  \setminus k} \alpha_{i}\Lambda^{t}_{n,i}(\tau-1)  + \\
		&\;\;\;\;\;
		\sum_{i \in \mathcal{M} \setminus i^*}^{}\sum_{j \in \mathcal{M} \setminus j^*}^{}2\beta_{ij}x_{ij}(\tau-1) \Lambda^f_{(i,j)}  + 	\sum_{j \in \mathcal{M} \setminus j^*}^{}\sum_{n \in \mathcal{N}  \setminus k} \gamma_j \Lambda^{c}_{n,j}(\tau-1)   \Big\}
		\\&\;\;\;\;\; - 
		\Big\{ \Big(\sum_{n \in \mathcal{N}  \setminus k}\alpha_{i^*}\hat{\Lambda}^t_{n,i^*}(\tau)+\psi_k \hat{\Lambda}^t_{k,i^*}(\tau)\Big) 
		\\&\;\;\;\;\;+ \Big(2\beta_{i^*j^*} x_{i^*j^*}(\tau-1)\hat{\Lambda}^f_{(i^*,j^*)} +2\psi_k\hat{\Lambda}^f_{(i^*,j^*)}\Big)  \\&\;\;\;\;\; +
		(\sum_{n \in \mathcal{N}  \setminus k}\gamma_{j^*}\hat{\Lambda}^c_{n,j^*}(\tau)+\psi_k\hat{\Lambda}^c_{k,j^*}(\tau)) \Big\} + \psi_k T_k \\&
		=  
		- \Big\{    \Big(\sum_{i \in \mathcal{M}}\sum_{n \in \mathcal{N}  \setminus k}\alpha_i\Lambda^t_{n,i}(\tau-1)-\sum_{n \in \mathcal{N}  \setminus k}\alpha_{i^*}\Lambda^t_{n,i^*}(\tau-1)\Big) 
		\\&\;\;\;\;\;+\Big(\sum_{i \in \mathcal{M}}\sum_{j \in \mathcal{M}}2\beta_{ij}x_{ij}(\tau-1)\Lambda^f_{(i,j)}  - 2\beta_{i^*j^*}x_{i^*j^*}(\tau-1)\Lambda^f_{(i^*,j^*)}\Big)  \\&\;\;\;\;\;  +
		\Big(\sum_{j \in \mathcal{M}}\sum_{n \in \mathcal{N}  \setminus k}\gamma_j\Lambda^c_{n,j}(\tau-1)-\sum_{n \in \mathcal{N}  \setminus k}\gamma_{j^*}\Lambda^c_{n,j^*}(\tau-1)\Big)
		\Big\}  
		\\&\;\;\;\;\;-
		\Big\{ 
		\Big(\sum_{n \in \mathcal{N}  \setminus k}\alpha_{i^*}\hat{\Lambda}^t_{n,i^*}(\tau)+\psi_k\hat{\Lambda}^t_{k,i^*}(\tau)\Big)  \\&\;\;\;\;\; +
		2\Big(\beta_{i^*j^*}x_{i^*j^*}(\tau-1)\hat{\Lambda}^f_{(i^*,j^*)} +\psi_k\hat{\Lambda}^f_{(i^*,j^*)}\Big) 
		\\&\;\;\;\;\;+ \Big(\sum_{n \in \mathcal{N}  \setminus k}\gamma_{j^*}\hat{\Lambda}^c_{n,j^*}(\tau)+\psi_k\hat{\Lambda}^c_{k,j^*}(\tau)\Big)
		\Big\} 
		+ \psi_k T_k
	\end{aligned} 
\end{equation}

\end{document}